\documentclass[abstract, a4paper, 12pt]{scrartcl}

\usepackage[utf8]{inputenc}
\usepackage[T1]{fontenc}
\usepackage{mathtools}
\usepackage{amsfonts}
\usepackage{amssymb}
\usepackage{amsthm}

\usepackage{soul}
\usepackage[normalem]{ulem}

\usepackage{graphicx}
\usepackage{subcaption}
\usepackage[english]{babel}
\usepackage{mathtools}

\usepackage{booktabs, multirow}

\usepackage{hyperref}
\hypersetup{
colorlinks,
linkcolor={blue!60!black},
citecolor={blue!60!black},
urlcolor={blue!60!black},
linktoc=page}

\usepackage{csquotes,doi}
\usepackage[
backend=biber,
style=alphabetic,
sorting=nyt]{biblatex}
\usepackage{authblk} 

\usepackage{todonotes}
\newcommand{\cL}{\mathcal{L}}

\newcommand{\cN}{\mathcal{N}}
\newcommand{\cT}{\mathcal{T}}

\renewcommand{\d}{\mathrm{d}}

\newcommand{\R}{\mathbb{R}}
\newcommand{\Z}{\mathbb{Z}}

\newcommand{\bbS}{\mathbb{S}}

\newcommand{\vb}[1]{\mathbf{#1}}

\newcommand{\dd}[2][]{%
  \ensuremath{\d^{#1}{#2}}
}

\newtheorem{theorem}{Theorem}
\newtheorem{definition}{Definition}
\newtheorem{proposition}{Proposition}
\newtheorem{lemma}{Lemma}
\newtheorem{corollary}{Corollary}
\newtheorem{conjecture}{Conjecture}

\theoremstyle{definition}

\def\bbR{\mathbb{R}}

\def\bbN{\mathbb{N}}

\def\bbS{\mathbb{S}}

\def\cA{\mathcal{A}}
\def\cN{\mathcal{N}}

\def\cT{\mathcal{T}}

\def\cL{\mathcal{L}}

\title{Chaotic light scattering \\ around extremal black holes}

\author[1]{Martijn Kluitenberg}
\author[2]{Diederik Roest}
\author[1]{Marcello Seri}
\affil[1]{Bernoulli Institute for Mathematics, Computer Science and Artificial Intelligence, \authorcr\small University of Groningen, The Netherlands,
\authorcr\normalsize \texttt{m.kluitenberg@rug.nl}, \texttt{m.seri@rug.nl}}
\affil[2]{Van Swinderen Institute for Particle Physics and Gravity, \authorcr\small University of Groningen, The Netherlands,
\authorcr\normalsize \texttt{d.roest@rug.nl}}

\date{\today}

\begin{document}

\maketitle
\abstract{We show that the scattering of light in the field of $N\geq 3$ static extremal black holes  is chaotic in the planar case. The relativistic dynamics of such extremal objects reduce to that of a classical Hamiltonian system. Certain values of the dilaton coupling then allow one to apply techniques from symbolic dynamics and classical potential scattering. This results in a lower bound on the topological entropy of order $\log(N-1)$, thus proving the emergence of chaotic scattering for $N \geq 3$ black holes.\\

}

\section{Introduction}

The highly appealing and impressive observations of the shadow of the M87 galaxy's central black hole by the Event Horizon Telescope \cite{EHT} (see the left panel of Figure~\ref{fig:bh}) has triggered a renewed interest in the scattering of light off black holes. In particular, it raises interesting questions regarding the possible emergence of chaos in such scattering processes. We address this problem for multi-black hole configurations from a mathematical perspective, focussing on the case of so-called extremal configurations. 

\begin{figure}
     \centering
     \includegraphics[width=0.47\textwidth,trim={0cm 1.35cm 0cm 1.15cm}, clip]{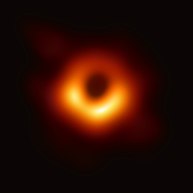}
    \hfill
     \includegraphics[width=0.47\textwidth]{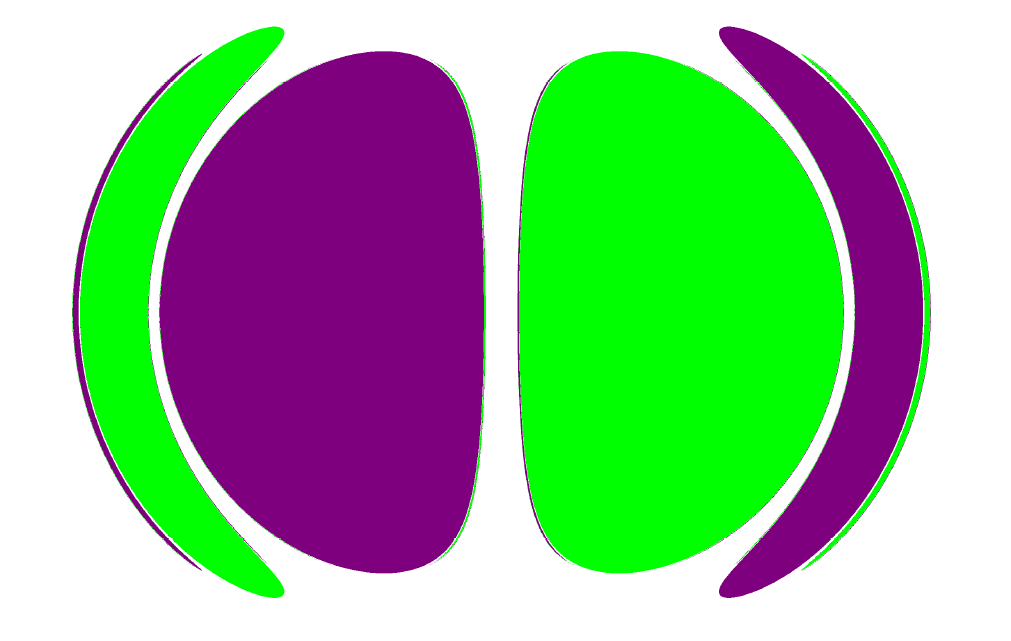}
    \caption{\textit{Left: image of the M87 galaxy's central black hole by the Event Horizon Telescope~\cite{EHT}.  Right: shadow of a pair of extremal black holes of equal mass~\cite{SD16}.}}
    \label{fig:bh}
\end{figure}

Starting in the non-relativistic regime, the scattering of test particles in the classical two-center problem is integrable and hence has no chaos. Moreover, it is well-known that scattering is chaotic in the case of $N \geq 3$ centers. Generic relativistic corrections, however, modify this classical story. For instance, replacing the classical centers with Schwarzschild black holes, the relativistic corrections induce chaotic behaviour in the scattering of light rays already for $N=2$, as already argued for in \cite{Con90} based on numerical calculations.

In this article, instead, we study the scattering of light rays in a space-time with $N$ extremal black holes. These are a special type of charged Reissner-N\"ordstrom black holes (in contrast to neutral Schwarzschild ones). In the extremal limit, where the mass and charge of the black hole are identified, these only exert velocity-dependent forces on each other. As a consequence, such extremal black holes allow for a static configuration. This suggest that the system can be viewed as the relativistic counterpart of the classical $N$-center problem; indeed we will outline precisely how to make such an identification. 

To give more context to the reader, we consider the \emph{Einstein-Maxwell-dilaton theory} (described in Section \ref{sec:EM} or \cite{TTC}). This is a relativistic theory described by the action
\[
S = \int_M \left( R - \partial_\mu \phi \partial^\mu \phi - e^{- 2 a \phi} F_{\mu \nu} F^{\mu \nu} \right) \sqrt{-g} \dd[4]{x},
\]
where $\phi$ is the so-called \emph{dilaton field}, and $a \geq 0$ is a coupling constant. There is a family of solutions to this theory depending on the parameter $a$, which may be interpreted as extremally charged black holes \cite{HH72s}. The metric for these solutions has the form
\begin{align*}
    \dd{s}^2
        &= - U^{-\frac{\alpha}{2}} \dd{t}^2 + U^{\frac{\alpha}{2}} (\dd{x^2} + \dd{y}^2 + \dd{z}^2)
\end{align*}
for a function $U$ to be defined later in \eqref{eq:defnU} and $\alpha$ defined as
\begin{equation}\label{def:alpha}
\alpha = \frac{4}{1 + a^2}.
\end{equation}
Since $a \geq 0$, this parameter takes values in the interval $(0,4]$. In the absence of external forces, light rays will follow the null-geodesics of the metric. Hence, their motion is described by zero-energy solutions of a Hamiltonian system.

There have been interesting formal attempts to show that the scattering of light is chaotic for this extremal relativistic system in the planar two-center. The authors of \cite{TTC} argue that the motion remains regular for all parameters\footnote{For $\alpha = 1$, see Corollary~\ref{cor:equivalence-kepler}, the relativistic problem becomes identical to the classical $N$-center Kepler problem, and hence the dynamics is not chaotic for $N=2$. Although we could not find any specific reference in which this result is explicitly mentioned, it is already implicitly assumed in the paper \cite{TTC}.} $0 \leq \alpha \leq 1$ while for $\alpha \in (1,4]$ the scattering should be chaotic. The same problem was also investigated in \cite{SD16} and extended to the non-planar case; see the right panel of Figure~\ref{fig:bh}. The historical claims are summarized in Table~\ref{tab:comparison} where we also relate the parameter $\alpha$ used in this paper with the more commonly used parameter $a$. 

\begin{table}[hb]
\centering
\begin{tabular}{@{}llllll@{}}
\toprule
Reference & Claim & Black hole  & $a$     & $\alpha$     & Configuration       \\ \midrule
\cite{Con90} & Chaotic & Neutral & $a = 0$ & $\alpha = 4$ & Planar \\
\begin{tabular}[c]{@{}l@{}}\cite{TTC} \\ ~\end{tabular} &
  \begin{tabular}[c]{@{}l@{}}Chaotic \\ Regular\end{tabular} &
\begin{tabular}[c]{@{}l@{}}Extremal \\ ~\end{tabular} &
  \begin{tabular}[c]{@{}l@{}}$0 \leq a < \sqrt{3}$\\ $a \geq \sqrt{3}$\end{tabular} &
  \begin{tabular}[c]{@{}l@{}}$1 < \alpha \leq 4$\\ $\alpha \leq 1$\end{tabular} &
  \begin{tabular}[c]{@{}l@{}} Planar \\ ~\end{tabular} \\
\cite{SD16}  & Chaotic & Extremal & $a = 0$ & $\alpha=4$   & Non-planar     \\ \bottomrule
\end{tabular}
\caption{Comparison of related claims in the physical literature regarding the scattering off black hole binary systems.}
\label{tab:comparison}
\end{table}

The main aim of this manuscript is to provide rigorous grounds for the claim of chaotic scattering for $N\geq 3$ and specific dilaton couplings.
\begin{theorem}\label{thm:main}
    Let $N \geq 3$ denote the number of extremal black holes in the planar Einstein-Maxwell-dilation theory as defined in Section~\ref{sec:EMD}, and assume that no combination of three such black holes is collinear.
    Then light scattering as defined in Section~\ref{sec:EM-light} is chaotic for all 
    \[
    \alpha_n = \frac{2n}{n+1} \in  [1,2), \quad n\in\bbN, 
    \]
    in the following sense: the relativistic flow is topologically semi-conjugate to a chaotic shift map on sequences on the alphabet $\cA=\{1,\ldots,N\}$ where the same symbols cannot appear twice in a row.
\end{theorem}
\noindent
As a consequence, the scattering system inherits chaotic properties from the shift map.

Extending the result to the planar two-center case, as well as to non-planar configurations (see also \cite{SD16}), is technically more involved and will be postponed for future research. We, however, believe that the result keeps holding true provided that $\alpha \neq 1$.
\begin{conjecture}
    Theorem \ref{thm:main} holds for all $\alpha\in(1,2)$, $N\geq2$ and in non-planar configurations and higher space dimensions.
\end{conjecture}
\noindent
One can readily observe that the range of values of $\alpha$ covered by previous claims on the $N=2$ case is much larger than our conjecture. It is worth to have a brief discussion of why that is the case. \\ 

Our aim for this manuscript is to show that we can apply ideas from classical hamiltonian scattering theory to the dynamics of light in the field of extremal Reisser-N\"ordstrom relativistic black holes, even though the latter are relativistic objects.
To this end, we need to ensure that the classical scattering operator, a map between the far past asymptotic of the trajectories to their far future, is well defined. This will depend both on the behaviour of the classical potential at infinity and its behaviour around the singularities. 

The decay of the potential at infinity heavily influences how we define the asymptotic velocities and the choice of the reference hamiltonian for the scattering comparison. Fortunately for us, scattering theory of long-range potentials is well developed~\cite{KnaufCM, Fejoz_2021} and our potential falls in the class.
However, to be able to compare the asympotic behaviour of the flows, the motion needs to be well defined at all times.
While this is not a problem for regular potentials, some work is needed to check regularizability of the singular ones, in the sense that we need to prove that they have a unique completion to all times.

In the case of $\alpha \geq 2$, it is a well known fact that the measure of initial conditions leading to a collision with the singularities in finite-time is positive, which greatly complicates the analysis and is our primary reason to avoid this range.
Moreover, while the flow of regular long-range regular potentials is always complete, our approach to regularization of the flow for our singular potentials will introduce a limitation on the allowed values of the parameter $\alpha\in (1,2)$.

We believe that this limitation can be overcome by using a different regularization technique, on the line of \cite{KK92}, using Jacobi metric to lift the problem to study a covering by Riemann surfaces. This ends up being a lengthier and more technical approach and we decided to postpone it for future research.
It is worth mentioning, though, that that the curvature of the Jacobi metric for the two-center relativistic planar case remains negative, which is a strong indicator that the result should hold true also for $N=2$ and hints at the fact that this could be a promising approach.\\

Of course, this paper is only scratching the surface of the potentials from cross-pollination between Hamiltonian mechanics one the one hand, and certain relativistic theories on the other. A look at the physical literature around extremal black holes shows interesting investigations of the same system in a number of different directions, see e.g. \cite{Shiraishi_1993,Assump_o_2018}, and there are many more questions and problems that are omitted here but would warrant a further look. The case of massive particles and the non-planar scattering are just the two most immediate ones.
In this respect, we hope for this work to further stimulate fruitful collaborations between mathematicians and physicists.\\

The paper is structured as follows. In Section \ref{sec:proof}, we prove Theorem \ref{thm:main} and we compute a bound on the topological entropy. The main idea of the proof is to rewrite the Hamiltonian for geodesic motion in spacetime as an autonomous Hamiltonian of the form ``kinetic energy $+$ potential energy'' to be able to apply technical tools from \emph{classical mechanics} and chaotic scattering. The required background in general relativity and classical scattering is reviewed in Sections $\ref{sec:EM}$ and $\ref{sec:scattering}$ respectively. \\

\textbf{Acknowledgements}.
We would like to thank Andreas Knauf and Nikolay Martynchuk for useful discussions in the development of this work.
MS research is supported by the NWO project 613.009.10 and MK and MS research is supported by NWO project OCENW.KLEIN.375.

\section{Relativistic models of $N$-center problems}\label{sec:EM}

\subsection{Review of General Relativity}

In this section we recall some of the basics of general relativity (GR) in order to understand the construction of the $N$-center solutions that we are interested in. As a general reference for this section, we refer the reader to \cite{WaldGR}. In GR, spacetime is a 4-dimensional manifold $M$ equipped with a Lorentzian metric $g_{\mu \nu}$ which we assume has signature $(-1,1,1,1)$. The simplest example is $\R^4$ with the metric $\eta = \mathrm{diag}(-1,1,1,1)$ or, analogously,
\[
\dd{s}^2 = -\dd{t}^2 + \dd{x}^2 + \dd{y}^2 + \dd{z}^2.
\]
The length of a parametrized curve $q: (a,b) \to M$ is given by
\[
\ell(q) = \int_a^b \| \dot{q} (\lambda) \| \dd{\lambda} = \int_a^b \sqrt{g_{\mu \nu} (q (\lambda)) \dot{q}^\mu (\lambda) \dot{q}^\nu(\lambda) } \dd{\lambda},
\]
which is invariant under reparametrization. Note that $\ell(q)$ may be zero, even if it connects two different points of spacetime; massive particles travel along curves with $\|\dot{q}(\lambda)\| > 0$ while for massless particles $\|\dot{q}(\lambda)\| =0$.  \\

The metric $g_{\mu \nu}$ is a dynamical variable, subject to equations of motion. These follow from a variational principle that minimizes the so-called \emph{Einstein-Hilbert action}
\[
S = \int_M \sqrt{-g} R \dd[4]{x},
\]
where $g = \det(g_{\mu\nu})$ (which is negative because of the signature), $R$ is the Ricci scalar associated to the metric, and, to simplify the notation, we have set the factor $\frac{c^3}{16 \pi G}$ to $1$.
The corresponding equations of motion are 
\begin{equation}\label{eq:fieldeqns}
    G_{\mu \nu} := R_{\mu \nu} - \frac{1}{2} R g_{\mu \nu} = 0,
\end{equation}
which are referred to as the \textit{(free space) Einstein field equations}. A \textit{solution} to GR is any metric $g_{\mu \nu}$ obeying (\ref{eq:fieldeqns}). \\

We now want to include other interactions into GR, such as electromagnetism. This is done by adding new terms to the Lagrangian density $\cL = \sqrt{-g}R$ which occurs in the action. In case of electromagnetism, we add a new dynamical variable $A_\mu$, which is called the 4-potential. We moreover define the field strength tensor
\begin{equation}
    F_{\mu \nu} = \nabla_\mu A_\nu - \nabla_\nu A_\mu.
\end{equation}
The action is then expressed in terms of $F_{\mu \nu}$ as
\begin{equation}
    S = \int_M \left(  R - \frac{1}{4} F_{\mu \nu} F^{\mu \nu} \right) \sqrt{-g} \dd[4]{x},
\end{equation}
A \textit{solution} to the above theory should now be stationary w.r.t. variations of both $g_{\mu \nu}$ and $A_\mu$. This yields two coupled equations
\begin{align}
\begin{cases}
    G_{\mu \nu} = \frac{1}{2} T_{\mu \nu}\\
    \nabla_\mu F^{\mu \nu} = 0
    \end{cases},
\end{align}
where $T_{\mu \nu}$ is the electromagnetic stress-energy tensor
\begin{equation}
    T_{\mu \nu} = g^{\rho \sigma} F_{\mu \rho} F_{\nu \sigma} - \frac{1}{4} g_{\mu \nu} F^{\rho \sigma} F_{\rho \sigma}.
\end{equation}
The two equations above are called the \textit{Einstein-Maxwell equations}.

\subsection{The Majumdar-Papapetrou solution}

We look for special solutions to the Einstein-Maxwell equations on a subset $\Omega \subseteq \bbR^4$. Consider a metric of the form
\[
\dd{s}^2 = - \frac{1}{U^2} \dd{t}^2 + U^2 (\dd{x}^2 + \dd{y}^2 + \dd{z}^2),
\]
where $U\equiv U(x,y,z)$ is a function of the spatial variables $x,y,z$ only. Taking
\[
A = \frac{1}{U} \dd{t},
\]
we obtain a solution to the Einstein-Maxwell equations if $U$ is a harmonic function.
This is very convenient, since the Laplace equation $\Delta U = 0$ is linear, while the Einstein field equations are fully nonlinear. In this paper we focus on the case 
\begin{equation}\label{eq:RNU}
    U = 1 + \sum_{i = 1}^N \frac{M_i}{r_i},
\end{equation}
with $r_i = \|\vb{x} - \vb{s}_i\|$ the distance to the fixed $i$\textsuperscript{th} center $\vb{s}_i$ in space. \\

The solution above was first derived by Majumdar \cite{Maj47} and Papapetrou \cite{Pap45}. It was later interpreted by Hartle and Hawking \cite{HH72s} as a spacetime with $N$ \emph{extremally charged} black holes; the different centers  have \emph{charges} $Q_i = M_i$.  Due to the extremal charges, we have a cancellation between the electrostatic repulsion and the gravitational attraction (more precisely, that there are only velocity-dependent forces), allowing for a static configuration of multiple black holes. Finally, it is impossible to add further charge without introducing e.g.~naked singularities; a discussion of further physical properties of the black holes can be found in \cite{kluitenberg2021chaotic}.

\subsection{Including dilaton charge}\label{sec:EMD}
We can further generalize the previous solution by including a scalar field $\phi$, called the dilaton field. Ignoring factors of $c$ and $G$, the action can be written as \cite{TTC}
\[
S = \int_M \left( R - \partial_\mu \phi \partial^\mu \phi - e^{- 2 a \phi} F_{\mu \nu} F^{\mu \nu} \right) \sqrt{-g} \dd[4]{x},
\]
where $a \geq 0$ is a coupling constant. We have omitted the dimensional constant in front of $R$ in order to be consistent with \cite{TTC}. The term $\partial_\mu \phi \partial^\mu \phi$ can be interpreted as kinetic energy. From the action, we see that the particles corresponding to $\phi$ are massless and have spin 0. \\

A \textit{solution} to the above theory then has to be stationary w.r.t. variations of $g_{\mu \nu}, A_\mu$ and $\phi$, giving three coupled equations. Since we do not need the explicit forms of the equations, we have chosen to omit them here. With $\alpha$ given by \eqref{def:alpha}, the previous solution generalizes to
\begin{equation}\label{eq:metrica}
    \dd{s}^2 = - U^{-\frac{\alpha}{2}} \dd{t}^2 + U^{\frac{\alpha}{2}} (\dd{x^2} + \dd{y}^2 + \dd{z}^2)
\end{equation}
together with 
\[
A = \frac{2}{\sqrt{\alpha}\; U} \dd{t}
\quad
\mbox{and}
\quad
e^{- \phi} = U^{\frac{a}{1 + a^2}}.
\]
The function $U$ is modified to
\begin{equation}\label{eq:defnU}
    U = 1 + \sum_{i = 1}^N \frac{4 M_i}{\alpha r_i},
\end{equation}
which reduces to the previous solution in case $\alpha = 4$, i.e. $a = 0$. \\

Like before, we interpret this solution as a configuration of \emph{extremal black holes}, satisfying the relation $Q_i^2 = M_i^2 + \Sigma_i^2$ ensuring the cancellation of velocity-independent forces, where $\Sigma_i$ is the dilaton charge. The Einstein-Maxwell-dilaton theory has several interesting special cases. One example is the case $\alpha = 1$, which is related to a reduction of Einstein gravity from 5 to 4 spacetime dimensions. This is referred to as Kaluza-Klein theory \cite{ravndal2004scalar}. As we have mentioned, $\alpha = 1$ is also special because in this case the motion of light becomes integrable. Other interesting limits are discussed in \cite{TTC, kluitenberg2021chaotic}.

\subsection{Motion of light}\label{sec:EM-light}

We are interested in studying the motion of light rays in a geometry described by the metric (\ref{eq:metrica}). Note that we only need to study $g_{\mu \nu}$, since photons do not have electromagnetic or dilaton charge. In the absence of external forces, particles will follow the geodesics of $g_{\mu \nu}$. These are determined as critical points of the length functional
\[
\ell(q) = \int_a^b \| \dot{q}(\lambda) \| \dd{\lambda}. 
\]
The corresponding Lagrangian is given by
\begin{equation}
    L = \frac{1}{2} g_{\mu \nu} (q (\lambda)) \dot{q}^\mu (\lambda) \dot{q}^\nu (\lambda). 
\end{equation}
Moreover, since photons are massless, we should have
    $p_\mu p^\mu = m^2 = 0$,
where $p_\mu = g_{\mu \nu} \dot{q}^\nu$
is the generalized momentum conjugate to $q^\mu$.
As an immediate consequence, light rays move along solutions to the Euler-Lagrange equations with $L = 0$. \\

We will now study the Hamiltonian formulation of the problem.
The Hamiltonian can be computed via Legendre transform and is given by
\begin{equation}
    H = \frac{1}{2} g^{\mu \nu} (q) p_\mu p_\nu = \frac{1}{2} \left( -U^{\frac{\alpha}{2}}(\vb{x}) p_t^2 + U^{- \frac{\alpha}{2}}(\vb{x}) (p_x^2 + p_y^2 + p_z^2) \right),
\end{equation}
where $g^{\mu \nu}$ with upper indices denotes the inverse of the metric tensor. The motion of photons takes place on the $\{ H = 0\}$ level set. Since the function $U$ depends only on the spatial variables, $t$ is a cyclic variable and $p_t$ is conserved. By reparametrizing the trajectory, we may set $p_t = -1$ without loss of generality. \\

Since we are on the level set $\{ H = 0 \},$ we can multiply the Hamiltonian by the nowhere vanishing function $U^{\frac{\alpha}{2}},$ which is equivalent to doing a time rescaling. We end up with the new Hamiltonian
\begin{equation}\label{eq:H'full}
    \widetilde H(x,y,z,p_x,p_y,p_z) = - \frac{1}{2} U^{\alpha}(\vb{x}) + \frac{1}{2} (p_x^2 + p_y^2 + p_z^2) + \frac{1}{2},
\end{equation}
which is in the form ``kinetic energy'' + ``potential energy'' with potential
\begin{equation}\label{eq:potential}
    V(\vb{x}) = \frac{1}{2} - \frac{1}{2} U^{\alpha}(\vb{x}) =  \frac{1}{2} - \frac{1}{2} \left( 1 + \sum_{i = 1}^N \frac{\widetilde  M_i}{r_i(\vb{x})} \right)^\alpha,
\end{equation}
where $\widetilde  M_i := \frac{4}{\alpha}M_i$, $r_i = r_i(\vb{x}) = \|\vb{x} - \vb{s}_i\|$.
The term $\frac{1}{2}$ is added in $\widetilde  H$ to make the potential $V(\vb{x})$ go to zero when $|\vb{x}| \to \infty.$ 

We consider the scattering of particles in a system with Hamiltonian (\ref{eq:H'}) at energy $E = \frac{1}{2}$. Notice that for $\alpha = 1$, we just obtain the Kepler Hamiltonian. The range of physically interesting parameters is $0 < \alpha \leq 4$. As announced in the introduction, we will restrict our attention to the planar case, i.e. we assume that all the centers are located in the $(x,z)$-plane. We can then set $p_y = 0$ without loss of generality, and we only have two degrees of freedom left. The Hamiltonian becomes:
\begin{equation}\label{eq:H'}
    \widetilde H (x,z,p_x,p_z) = \frac{1}{2} (p_x^2 + p_z^2) + V(x,z),
\end{equation}
with
$V$ defined in \eqref{eq:potential}. \\

We have just proven the following statement:
\begin{proposition}\label{prop:eq-to-cl-ham}
    Motion of light in the Einstein-Maxwell-dilaton theory under the influence of $N$ extremal black holes is described, up to a time-reparametrization, by the evolution of the classical Hamiltonian $\widetilde{H}$ from \eqref{eq:H'full} at energy $\frac12$.
    The planar case, further reduces to the classical Hamiltonian \eqref{eq:H'}.
\end{proposition}

\begin{corollary}\label{cor:equivalence-kepler}
    Motion of light in the Einstein-Maxwell-dilaton theory under the influence of $N$ extremal black holes with dilaton coupling $\alpha=1$ is described, up to a time-reparametrization, by the evolution of a classical $N$-center Kepler Hamiltonian.
\end{corollary}
\noindent
In the next section, we briefly develop the mathematics of potential scattering, which will then be applied to this specific potential.

\section{Classical potential scattering}\label{sec:scattering}

The main result of this paper is that scattering of light by $N$ extremal black holes can be completely described in terms of scattering by a classical potential. We review here the minimal amount of the necessary classical scattering theory, for additional details one can refer to \cite[Chapter 12]{KnaufCM}.

\subsection{Knauf's scattering degree}

Consider a $d$-dimensional configuration space $\R_q^d$. Conjugate momenta are elements of the cotangent space $T_q^*\R^d \simeq \R_p^d$, and hence the phase space is $P = \R_q^d \times \R_p^d$. We will sometimes remove one or more points from the configuration space. In such cases, the phase space will be $P = \left( \R_q^d \setminus \{ s_1,\dots,s_n \} \right) \times \R_p^d$. Below, we omit boldface notation from the vectors. \\

Given a potential $V\in C^2(\R_q^d,\R)$, we will define the corresponding mechanical Hamiltonian and the \emph{free Hamiltonian} respectively as
\begin{equation}
	H(p,q) = \frac{1}{2} \|p\|^2  + V(q),
\quad\mbox{and}\quad
	H_0(p,q) = \frac{1}{2} \|p\|^2.
\end{equation}
The Hamiltonian flows will be denoted by $\Phi_t$ and their projections onto the configuration space $q(t;x)$ and $p(t;x)$, $x$ denoting the initial condition. An extra suffix $0$ will be added to the free flows. The continuity restriction on the potential is added in order to guarantee existence and uniqueness of solutions.

The asymptotic behaviour of the potential at infinity plays an important role and requires to introduce some further terminology.
\begin{definition}
	Let $V \in C^2 (\R_q^d , \R)$ and $\langle q \rangle := \sqrt{\|q\|^2 + 1}$. 
	We say that $V$ is
	\begin{itemize}
		\item \emph{long range} if, for the appropriate $\epsilon > 0$,
		\[
		|\partial^i V(q)| \leq c \langle q \rangle ^{- |i| - \epsilon} \quad \forall q \quad  \forall |i| \leq 2 
		\]
		\item \emph{short range} if, for the appropriate $\epsilon > 0$, 
		\[
		|\partial^i V(q)| \leq c \langle q \rangle ^{-|i|-1-\epsilon} \quad \forall q \quad  \forall |i| \leq 2
		\]
		\item \emph{compactly supported} if $V$ is smooth and vanishes outside of some compact set $K \subseteq \R_q^d.$
	\end{itemize}
\end{definition}
\noindent
Intuitively, a potential is long range if its force falls off faster than $1/r$, and it is short range if its force falls of faster than $1/r^2$. Obviously, compactly supported potentials are short range, and short range potentials are long range. \\

Assume that $V$ is long range. Under this assumption, for $d \geq 2$ the limit $\lim_{\|q\| \to \infty} V(q)$ exists. By adding a suitable constant, we can and will \textbf{always} assume that this limit equals zero. We define the following invariant subsets of the phase space $P$:
\begin{definition}
	Let $V$ be a long range potential.
	\begin{itemize}
		\item Define the sets
		\[
		b^\pm = \left\{ x \in P : \limsup_{t \to \pm \infty} \| q(t;x) \| < \infty  \right\}.
		\]
		We say that $x \in P$ is a \emph{bound state} if $x \in b^+ \cap b^-$. 
		\item Define the sets
		\[
		s^\pm = \left\{ x \in P : \lim_{t \to \pm \infty} \|q(t;x)\| = \infty  \right\}.
		\]
		We say that $x \in P$ is a \emph{scattering state} if $x \in s^+ \cap s^-$.
		\item Define the sets $t^{\pm} = b^{\pm} \cap s^{\mp}$.
		We say that $x \in P$ is a \emph{trapped state }if $x \in t^+ \cup t^-$. 
		\item For a real number $E$, we define $\Sigma_E = H^{-1}(E)$ to be the preimage of $E$ under the Hamiltonian, i.e. all states with energy $E$.
		The sets $b^{\pm}_E, b_E, s^\pm_E,s_E,t^{\pm}_E$ are defined as the intersection of $\Sigma_E$ with the respective set, e.g. $b_E = b \cap \Sigma_E.$
		\item We say that an energy $E$ is \emph{nontrapping} if $t_E = \emptyset$. The set of nontrapping energies is denoted by $\mathcal{N}\mathcal{T}.$ For such energies, the unbounded trajectories asymptotically look like straight lines in both forward and backward time.
	\end{itemize}
\end{definition}

In order to introduce the scattering degree we need to ensure that classical potential scattering is well defined, the following classical results guarantee that both for regular potentials and a relevant family of singular ones.

Let $V$ be at least long range, $C^2$ and vanishing at infinity and let $E > 0$ be a nontrapping energy.
The following are well-known facts, see e.g. Theorem 12.5, Corollary 12.8 and Theorem 12.11 from \cite{KnaufCM} or Section 2 of \cite{KK08} and the references therein.

\begin{lemma}
    For any scattering state $x \in s_E^\pm$ the \emph{asympotic directions} 
    \[
      \hat{p}^{\pm} : s^\pm \to \bbS^{d-1}, \quad
      \hat{p}^{\pm} (x) = \lim_{t \to \pm \infty} \frac{p(t;x)}{\sqrt{2E}}
    \]
    and the \emph{impact parameters}
    \[
        q^\pm_\perp : s^\pm \to \R^d, \quad
        q^\pm_\perp (x) = \lim_{t \to \pm \infty} \left( q(t;x) - \langle q(t;x),\hat{p}^\pm (x) \rangle \hat{p}^\pm(x) \right)
    \]
    are well defined continuous flow-invariant functions.
\end{lemma}

A crucial point for our analysis is that scattering in a classical potential at a fixed nontrapping energy $E > 0$ can be characterized by a topological invariant, introduced in \cite{KK92}, via a transformation that maps asymptotic trajectories in backwards time to those in forward time. The nontrapping condition is important here to ensure that our map will be defined everywhere. \\

By construction, the asymptotic impact parameter is orthogonal to the asymptotic direction. Hence, the combination $(q_\perp^\pm , \hat{p}^\pm)$ defines a point in the cotangent bundle $T^* \bbS^{d - 1}.$ Moreover, 
if two points are on the same orbit, then $(q_\perp^\pm , \hat{p}^\pm)$ has the same value on them. We will thus quotient out the action of $\Phi_{t}$ to obtain the following: 

\begin{theorem}{\cite[p. 6]{KK08}}
	Let $E \in \cN \cT.$ The maps $A^\pm_E : s^\pm_E / \Phi_\R \to T^* \bbS^{d - 1}$ are well-defined homeomorphisms. 
\end{theorem}

For a nontrapping energy, we define the map
\[
S_E = (Q_E,\hat{P}_E) = A_E^+ \circ (A_E^-)^{-1} :  T^* \bbS^{d - 1} \to  T^* \bbS^{d - 1}.
\] \noindent The map $\hat{P}_E : T^* \bbS^{d - 1} \to \bbS^{d- 1}$, which is the composition of $S_E$ with the canonical projection, is continuous by the previous theorem.  We consider its restriction to a single cotangent space $\hat{P}_{E,\theta} : T^*_\theta \bbS^{d - 1} \to \bbS^{d - 1}.$ This map gives a relation between an impact parameter $q_\perp$ as measured w.r.t. some reference angle $\theta$ and the corresponding asymptotic direction. 

\begin{lemma}{\cite[Lemma 2.5]{KK08}} \label{lem:limit}
	We have
	\[
	\lim_{\|q_\perp\| \to \infty} \hat{P}_{E,\theta} = \theta.
	\]\noindent
Thus, the map $\hat{P}_{E,\theta}$ extends uniquely to a continuous function $\hat{\mathbf{P}}_{E,\theta}: T^*_\theta \bbS^{d - 1} \cup \{ \infty \} \cong \bbS^{d - 1} \to \bbS^{d - 1}$.
\end{lemma} 
\noindent
For maps from the $(d - 1)$-sphere to itself, we can define the \textit{(topological) degree}. This is a concept from algebraic topology which generalises the winding number that directly leads us to the scattering degree.
\begin{definition}{\cite[p. 7]{KK08}}
We define the \emph{scattering degree}
\begin{align}
\deg: \cN \cT \to \Z, \qquad \deg(E) := \deg(\hat{\mathbf{P}}_{E,\theta}).
\end{align}
\end{definition}
\noindent 
Since the map $\hat{P}_{E,\theta}$ is continuous in both its input $q_\perp$ and its parameter $\theta$, the degree of $E$ does not depend on the choice of reference direction $\theta$. Moreover, it can be shown that $\deg$ is well-defined and locally constant on the set of nontrapping energies \cite[p. 7]{KK08}. \\

There are some cases in which we can explicitly work out the degree, the Theorem below summarizes \cite[Proposition 2-4]{Kna99} for our analysis. The assumption $d = 2$ below is made to simplify computations.

\begin{theorem}\label{thm:deg_known} Assume that $d = 2$. 
	\begin{itemize}
		\item Let $V$ be smooth and short range. Then, there is an energy threshold $E_{\cN\cT}$ above which the motion is nontrapping. For $E \geq E_{\cN\cT}$, we have $\deg(E) = 0.$ 
		
		\item For a centrally symmetric short range potential, if $0 < E < V_{\mathrm{max}}$ is nontrapping, then $\deg(E) = 1.$
		
		\item For $V(q) = -\frac{1}{|q|^\alpha}$ defined on $\hat{M} = \R^2 \setminus \{0\}$, the motion can be regularized if $\alpha = \frac{2n}{n + 1}$ for $n \in \mathbb{N}$. In this case, all positive energies are nontrapping and $\deg(E) = -n.$
	\end{itemize}
\end{theorem}
\noindent
We will only be concerned with the latter part of this theorem.
In the next section we will discuss how nontrivial degrees of scattering provide the key ingredient to track chaotic scattering.

As you can read in the latter statement in the theorem, and as we mentioned already in this chapter, in the presence of singularities we need to ensure regularizability of the motion. 
That is, we need to ensure that there is a unique way to extend the trajectories, which may in principle reach the singularity in finite time, to a well-defined trajectory over all times.
We will come back to this point also later on, when we will need to justify the regularizability of the motion for our relativistic case.

\subsection{Symbolic dynamics}

Our main goal is to prove the existence of chaotic scattering in the potential defined by (\ref{eq:potential}). Theorem \ref{thm:symb_dyn_deg} below serves as a sufficient condition. We briefly explain the main ideas. Consider an alphabet $\cA = \{a,b,c,\dots\}$ consisting of $N$ ``symbols''. It is well known that the shift map
\[
\sigma(s_1,s_2,\dots) = (s_2,s_3,\dots)
\]
defined on $\cA^\bbN$ is chaotic \cite{WigginsDS}, in the sense that we have topological transitivity, dense periodic points and sensitive dependence upon initial conditions. Here, we topologize $\cA^\bbN$ by setting 
\[
d(s,t) = \sum_{n = 1}^\infty \frac{d(s_n,t_n)}{2^n},
\]
where $d(s_n,t_n)$ is the discrete metric on $\cA$. To prove that a different dynamical system is chaotic, it suffices to show that it ``contains'' a shift map. We introduce a Poincar\'e section $\Lambda$ in the state space, and we define the first return map $P : \Lambda \to \Lambda$. We also define a function $F : \Lambda \to \cA$ for an appropriate choice of alphabet $\cA$. We can then associate to every state $x$ a symbolic sequence $(h(x))_n = F(P^n(x))$. The key step in proving that $h$ is a \emph{topological semi-conjugacy} is proving that every symbolic sequence is actually realized by an orbit in the continuous time system. To this end, we will use the topological degree.\\

Let us see how this works. Consider a potential $V = \sum_{i = 1}^N V_i$, where each $V_i$ is smooth and compactly supported in some ball
$B_i = \{ \vb{x} \in \R^2 : \|\vb{x} - \vb{s}_i\| < R_i \}.$
We will assume that the supports are \emph{disjoint}, and that they don't \emph{shadow} each other, meaning that no straight line meets more than two of the $B_i$. With the addition of some technical assumptions, we can introduce symbolic code based on the order in which a trajectory meets the balls. \\

The letters of the symbolic code for the case at hand are just $1,2,\dots,N.$ In this case, we say that a sequence is admissible if the same number never occurs twice in a row. The main technical tool is the following:

\begin{theorem}{\cite[Theorem 2]{Kna99}}\label{thm:symb_dyn_deg}
	Let $V$ defined as above with $N \geq 2$, $E$ nontrapping for each of the individual potentials $V_i$ and $\deg(V_i) \neq 0$ for all $i = 1,2,\dots,N$.
	For every interval $J_l^r = \{ j \in \Z : l \leq j \leq r  \}$, every admissible sequence $(k_j)$ and every $\hat{p}^\pm \in \bbS^1$ there is a trajectory with energy $E$ which meets the balls $B_{k_j}$, $i\in J_l^r$ precisely in the order prescribed by $(k_j)$. Moreover,
	\begin{itemize}
		\item If $l \neq - \infty$ then this trajectory in $s_E^-$ has initial direction $\hat{p}^-$. Otherwise it belongs to $b_E^-$.
		\item If $r \neq  \infty$ then this trajectory in $s_E^+$ has final direction $\hat{p}^+$. Otherwise it belongs to $b_E^+$.
	\end{itemize}
	In particular, $E$ is a trapping energy for $V$. 
\end{theorem}
\noindent
This theorem will allow us in the next section to prove the existence of chaos in relativistic $N$-center problems. If we can prove that each of the black holes, in isolation, provides a non-zero scattering degree and that we are at the right energy level, this theorem is telling us that we can achieve trajectories with an arbitrary number of detours around three (or more) non-collinear black holes.


\section{Proof of chaotic scattering of light}\label{sec:proof}

Proposition~\ref{prop:eq-to-cl-ham} allows us to reduce the problem to classical potential scattering for the Hamiltonian \eqref{eq:H'}.
In particular, we will study the planar motion of particles in the potential
\begin{equation}
    V(\vb{x}) =   \frac{1}{2} - \frac{1}{2} \left( 1 + \sum_{i = 1}^N \frac{\widetilde  M_i}{r_i(\vb{x})} \right)^\alpha,
\end{equation}
at $E = \frac{1}{2}$. For convenience, we omit the tilde in the following. 
What is left to do is to show that we can reduce the problem to one that can be described in terms of the theory presented in Section~\ref{sec:scattering}.\\

As a first step, let us study the form of $V(\vb{x})$ close to $\vb{x} = \vb{s}_i$.
Below we will drop the boldface vector notation from $\vb{x}$ and $\vb{s}_i$ when there is no possible confusion. \\

Suppose for the moment that there is only one center, located at $x = s$. A standard Taylor expansion gives: 
\begin{align*}
	V(x) &= \frac{1}{2} - \frac{1}{2} \frac{M^\alpha}{\| x - s\| ^\alpha} - \frac{1}{2} \frac{\alpha M^{\alpha - 1}}{\| x - s \| ^{\alpha - 1}} + O \left( \frac{1}{\| x - s \|^{\alpha - 2}} \right).
\end{align*}
For $\alpha \in (1,2)$, we have that $\alpha - 1 \in (0,1)$ and $\alpha - 2 \in (-1,0)$. Hence, only the first two terms in the Taylor expansion are singular, while all higher order terms are regular. By using Taylor's theorem with remainder, we can write
\[
V(x) = \frac{1}{2} - \frac{C}{\| x - s\| ^\alpha} -  \frac{D}{\| x - s \| ^{\alpha - 1}} + W(x)
\]
with $C,D$ positive constants and $W(x)$ a smooth function. 

Near $x = s$, the dominant contribution is $V(r) \propto - \frac{C}{r^\alpha} - \frac{D}{r^{\alpha-1}}$.
We will see below that for $\alpha\in (1,2)$, only the larger term will play a role and the dynamics will be effectively characterized by potentials of the form $V(r) \propto -\frac{1}{r^\alpha}$, which are the kind described in Theorem~\ref{thm:deg_known}.

\subsection{Scattering by leading potential terms}

For the moment, let us only consider the first term in the Taylor expansion above. Let $\psi_i : \bbR^2 \to \bbR$ be smooth cut-off functions supported in a neighbourhood $B_i = \{ x \in \bbR^2 : \|x - {s}_i \| < R_i \}$ of the $i$th singularity, such that $\psi_i \equiv 1$ in some open subset of $B_i$ containing $s_i$, e.g. $\{ x \in \bbR^2 : \|x - {s}_i \| < \frac34 R_i \}\subset B_i$, and $\partial_r \psi_i \leq 0$. The radii $R_i$ are to be chosen later. \\

We replace the original potential by 
\begin{equation}
	V(x) = - \sum_{i = 1}^N \frac{M^\alpha}{2 \|x - s_i\|^\alpha} \psi_i(x) = \frac{1}{2} \sum_{i = 1}^N V_i(x). 
\end{equation}
From Theorem \ref{thm:deg_known}, we know that for $\alpha = \frac{2n}{n + 1}$, $n \in \bbN$, the individual potentials without $\psi_i$ define a regularizable flow. Moreover, each positive energy is nontrapping and the degree equals $-n$. These properties persist upon multiplying by the cut-offs, as shown in the following Lemma:

\begin{lemma}
	For potentials of the form
	\begin{equation}
	V(x) = 	- \frac{M^\alpha}{2 \| x - s_i\|^\alpha} \psi_i(x)
	\end{equation}
	with $\alpha = \frac{2n}{n + 1}$ each positive energy is nontrapping and $\deg(E) = -n.$
\end{lemma}

\begin{proof}
	We first prove the result on the degree. Write
	\begin{align*}
	V(x) 
	&=  - \frac{M^\alpha}{2 \|x - s_i\|^\alpha} + \left[\frac{M^\alpha}{2 \|x - s_i\|^\alpha} (1 - \psi_i(x)) \right] \\
	&=  - \frac{M^\alpha}{2 \|x - s_i\|^\alpha} + W(x).
	\end{align*}
Note that the function $W$ in square brackets is smooth at $s_i$, since $1 - \psi_i$ is zero in some neighbourhood of $s_i$. Hence, it follows from \cite[Proposition 4.1]{KK08} that the degree is still $-n$.

To prove that $E$ is nontrapping, we use \cite[Remark 3.3]{KK08}. If the estimate $\langle x , \nabla V(x) \rangle \leq 0$ holds, then all positive $E$ which are regular values of $V$ are nontrapping. It is a straightforward computation to show that the estimate holds for $- \frac{1}{r^\alpha}$ potentials. Since the radial derivative of the cutoff is nonpositive, the previous inequalities remain true and all positive energies are nontrapping.
\end{proof}

In the next subsection, we consider also the second term in the expansion. 

\subsection{Scattering by the full potential}

When there are $N>2$ centers, we can perform the Taylor expansion near each center following our discussion at the beginning of the section. This gives
\begin{align*}
	V(x) &=  - \frac{1}{2} \left( 1 + \frac{4}{\alpha}\sum_{i = 1}^N \frac{M_i}{\|x - s_i\|} \right)^\alpha + \frac{1}{2} \\
	&= \sum_{i = 1}^N \left( - \frac{C_i}{\|x - s_i\|^\alpha} - \frac{D_i}{\|x - s_i\|^{\alpha-1}} + W_i(x) \right) + \mathrm{Cst.}
\end{align*}
where the potentials $W_i$ are smooth and $C_i, D_i \in \R$ are explicit constants.

Like in the previous subsection, we can use a smooth cut-off $\psi$ to reduce to the compactly supported case. Indeed,
\begin{align*}
	V(x) &= - \sum_{i = 1}^N  \psi(\|x - s_i\|)  \left( \frac{C}{\|x - s_i\|^\alpha} + \frac{D}{\|x - s_i\|^{\alpha - 1}} + W_i(x) \right) \\ 
	&\quad - \sum_{i = 1}^N (1 - \psi(\|x - s_i\|))  \left( \frac{C}{\|x - s_i\|^\alpha} + \frac{D}{\|x - s_i\|^{\alpha - 1}} + W_i(x) \right) + \mathrm{Cst.} \\ 
	&= - \sum_{i = 1}^N \psi(\|x - s_i\|) \left( \frac{C}{\|x - s_i\|^\alpha}  + \frac{D}{\|x - s_i\|^{\alpha - 1}} \right) + W(x) + \mathrm{Cst.}
\end{align*}
The term involving $(1 - \psi (\| x - s_i\|))$ is smooth, because the part in brackets is smooth except at $x = s_i$, where it is zero. Hence, this part can be absorbed into a single smooth function $W(x)$. \\

We now study the properties of
\begin{equation}\label{eq:cut-potential}
	V(r) = -\frac{C}{r^\alpha} - \frac{D}{r^{\alpha - 1}}
\end{equation}
with $C,D$ positive constants and $\alpha \in (1,2)$, which occurs in the equations above. We have shown earlier that, when $D = 0$ and $\alpha = \frac{2n}{n + 1}$, the motion can be regularized and the scattering degree is given by $\deg (E) = -n$. We will now show that since the term involving $- \frac{1}{r^{\alpha - 1}}$ is of lower order, it does not affect the scattering degree.\\

The regularization of the motion follows from a direct modification of the proof of \cite[Proposition 4.1]{KK08}.
Also in this case the pericentral radius, i.e. the minimal distance of the trajectory from the singularity, is well-defined and non-vanishing, and one can control the lower order singularity observing that $\frac{C}{\| x - s_i \|^\alpha} + \frac{D}{\|x - s_i \|^{\alpha - 1}} = \frac{C}{\| x - s_i \|^\alpha} ( 1 + \frac{D}{C} \|x - s_i \|)$. The fact that there is a well-defined limit in vanishing angular momentum for $\alpha_n$ is derived below.
Since the proof is completely analogous to \cite[Proposition 4.1]{KK08}, we will not rewrite all the details here. \\

Following \cite{Kna99, KK08} we obtain for the scattering angle
\[
\Delta \phi(E , \ell) = 2 \int_0^{v_{\mathrm{max}}} \frac{\dd{v}}{\sqrt{ E \ell^{\frac{2 \alpha}{2 - \alpha}}  2^{\frac{- \alpha}{2 - \alpha}} - v^2 + C v^\alpha + D v^{\alpha - 1} \left( \frac{\ell}{\sqrt{2}}\right)^{\frac{2}{2 - \alpha}}  }}  - \pi, 
\]
where now $v_{\mathrm{max}}$ is the smallest positive solution to the equation 
\[
E \ell^{\frac{2 \alpha}{2 - \alpha}}  2^{\frac{- \alpha}{2 - \alpha}} - v^2 + C v^\alpha + D v^{\alpha - 1} \left( \frac{\ell}{\sqrt{2}}\right)^{\frac{2}{2 - \alpha}} = 0.
\]
When $\ell \to 0$, notice that $v_{\mathrm{max}}$ converges to the solution of the simpler equation
\[
C v^\alpha = v^2 \implies v_{\mathrm{max}} = C^{\frac{1}{2 - \alpha }}. 
\]
Hence, upon interchanging the limit with the integral, we obtain
\[
\Delta \phi = 2 \int_0^{C^{\frac{1}{2 - \alpha}}} \frac{\dd{v}}{\sqrt{ C v^\alpha - v^2 }} = \frac{2 \pi}{2 - \alpha} - \pi,
\]
which does no longer depend on $D$. Moreover, the estimate
\[
\left\langle x ,  \nabla \left( -\frac{C}{\|x\|^\alpha} - \frac{D}{\|x\|^{\alpha - 1}} \right) \right\rangle = -C \left \langle x, \nabla \left( \frac{1}{\|x\|^\alpha} \right) \right\rangle - D \left \langle x, \nabla \left( \frac{1}{\|x\|^{\alpha - 1}} \right) \right\rangle \leq 0 
\]
shows that the potential remains nontrapping. Filling in $\alpha = \frac{2n}{n + 1}$ into the formula for $\Delta \phi$ then gives $\Delta \phi = n \pi$, hence the degree still equals $-n$. \\

Let us apply Theorem~\ref{thm:symb_dyn_deg}. Since we are assuming that no three centers lie on a single line, by shrinking the support of $\psi$, we can make sure that the supports of the potentials above do not shadow each other. It is now enough to observe that
\[
V(x) = - \sum_{i = 1}^N \psi(\|x - s_i\|) \left( \frac{C}{\| x - s_i \|^\alpha} + \frac{D}{\|x - s_i \|^{\alpha - 1}} \right) + W(x) + \mathrm{Cst}
\]
is \textit{identical} to our original potential modulo smooth perturbations. As in \cite[Section 4]{KK08}, the presence of the extra smooth terms does not affect the regularization nor the computation of the degree. \\

Thus, by Theorem \ref{thm:symb_dyn_deg}, we obtain a topological semi-conjugacy of the Hamiltonian flow of \eqref{eq:H'} to a shift map of sequence on the alphabet $\cA = \{ 1, 2 ,\dots, N\}$ without repeating symbols. We now bound the entropy of the above symbolic code in order to complete the proof.

\subsection{Positive topological entropy}

The topological entropy of symbolic systems $(X, \tau)$, where $\tau$ is a shift on a closed invariant collection $X$ of sequences, is defined as the limit
\[
S(\tau) := \lim_{k\to\infty} \frac{\log |W_k|}{k},
\]
where $|W_k|$ denotes the set of all words of length $k$ appearing anywhere in the sequences in $X$. \\

First of all, observe that our symbolic code is pruned: the same symbol may not appear twice in a row.
Secondly, we need to remember that we only have a semi-conjugacy, and therefore we can derive lower bounds but not exact identities.

For $N \geq 2$, to compute the entropy of our symbolic code, we look at all possible words of length $k$. Since the unique restriction is the impossibility of repetition of a symbol, there are $N (N - 1)^{k - 1}$ possibilities. Hence,
\[
S_{\mathrm{code}}(N) = \lim_{k \to \infty} \frac{\log \left( N (N  - 1)^{ k - 1}\right)}{k} = \log(N - 1).
\]
Then, for $N = 2$, $S_{\mathrm{code}}(2) = 0$, while for $N \geq 3$ we have $S_{\mathrm{code}}(N) > 0$. Since the original system contains a part which is conjugate to the shift map, we have that the topological entropy for light scattering is positive for $N \geq 3$.

The fact that for $N = 2$ the lower bound on the entropy is zero should not be surprising: the only allowed code in this case is $\dots121212\dots$. This is showing the limitations of our current construction since while this is true in the case $\alpha=1$, one can observe that for $\alpha >1$ the scattering angle $\Delta\phi$ can be larger than $\pi$: trajectories are allowed to cross the line between the centers and the code above can denote two rather different types of trajectories, see figure \ref{fig:orbits}. This fact, as also noted by \cite{Con90, Assump_o_2018} in the case $\alpha=4$, is a good indicator of the possibility of chaotic scattering for $N=2$ and $\alpha >1$.
\begin{figure}[ht!]
    \centering
    \emph{Classical two-centers Kepler problem, $\alpha=1$}\\
    \includegraphics[width=0.8\linewidth]{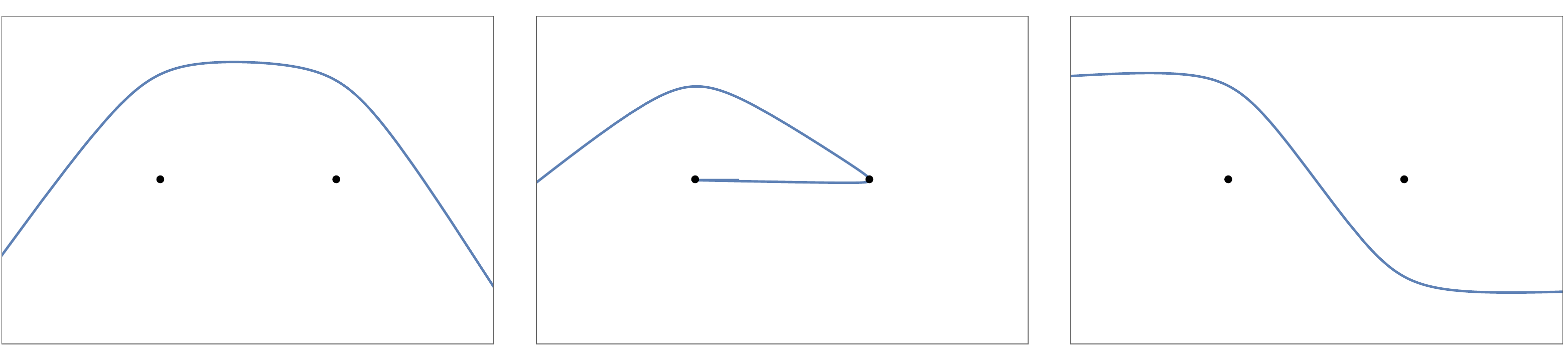}\\
    \emph{Relativistic two-centers problem with $\alpha=\frac32$}\\
    \includegraphics[width=0.8\linewidth]{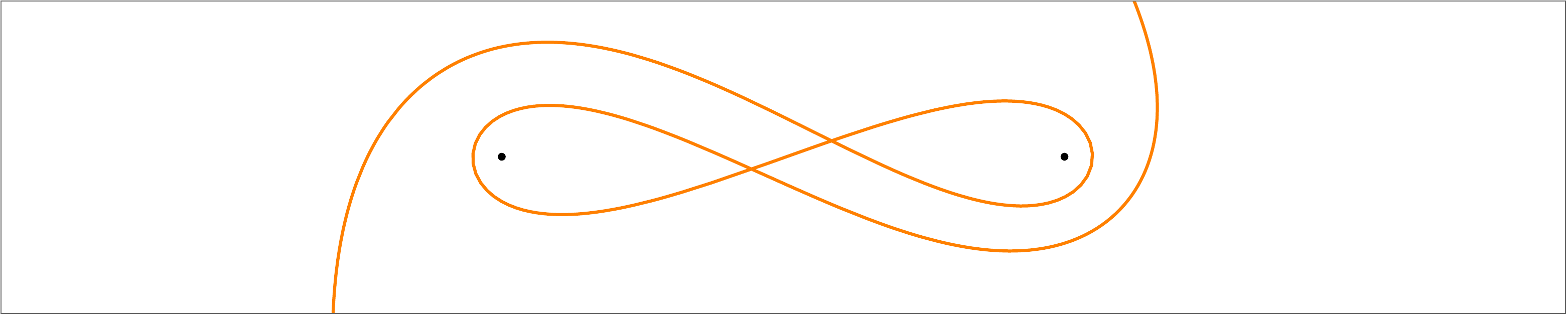}
    \caption{\textit{Upper panels: classification of the allowed scattering orbits for the classical two-centers Kepler problem ($\alpha=1$) for a fixed positive energy, from \cite{Seri_2015}. Lower panel: example of planar orbit with multiple crossings of the line connecting the centers for $\alpha = \frac{3}{2}$. In all plots the centers are marked with a black dot located at $(\pm 1,0)$.}}
    \label{fig:orbits}
\end{figure}

\section*{Conflict of interest}
On behalf of all authors, the corresponding author states that there is no conflict of interest.

\printbibliography

\end{document}